\documentclass[11pt,a4paper]{article}

\usepackage[margin=25mm]{geometry}
\usepackage[T1]{fontenc}
\usepackage{lmodern}
\usepackage{amsmath,amssymb,amsthm}
\usepackage{microtype}
\usepackage[hidelinks]{hyperref}
\numberwithin{equation}{section}

\newtheorem{proposition}{Proposition}
\newtheorem{corollary}{Corollary}
\theoremstyle{definition}
\newtheorem{definition}{Definition}
\newtheorem{example}{Example}
\theoremstyle{remark}
\newtheorem{remark}{Remark}

\newcommand{\id}{\operatorname{id}}
\newcommand{\tr}{\operatorname{tr}}

\title{Contact canonoid maps and their conserved and dissipated quantities}
\author{R. Azuaje\\
Department of Physics, Faculty of Nuclear Sciences and Physical Engineering,\\
Czech Technical University in Prague. Břehová 7, 115 19 Praha 1, Czech Republic}
\date{}

\begin{document}
\maketitle

\begin{abstract}
We study conserved and dissipated quantities associated with contact
canonoid transformations and their extension to possibly singular smooth
maps. The defining condition is imposed on the pulled-back contact
one-form, without requiring it to remain contact. A polynomial
construction produces first integrals from one-forms having the same
dissipation rate as the original contact form. Two normalizations extend
this construction to maps whose associated Hamiltonians have different
dissipation rates. On the set where the associated Hamiltonian is nonzero,
we construct a smooth invariant endomorphism whose traces
are conserved. This construction includes zeros of the
original Hamiltonian, where the endomorphism vanishes. We relate the tensor traces to the coefficients of the
Hamiltonian-normalized polynomial and obtain an upper bound of
$n$ functionally independent trace invariants in dimension
$2n+1$. A five-dimensional singular map attains this bound. Examples include
singular maps for the damped free particle, a transformation with unequal
dissipation rates, the damped harmonic oscillator, and motion under
gravity with linear friction.
\end{abstract}


\section{Introduction}

Contact Hamiltonian mechanics provides a geometric description of
dissipative systems in which the Hamiltonian generally satisfies a dissipation law rather than a conservation law
\cite{LL2019,LS2017,BCT2017,GGMRR2020,GLR2022}. Functions obeying the same dissipation law
as the Hamiltonian are called dissipated quantities \cite{BG2021,GGMRR2020,LL2020}. Their quotients,
where defined, are constants of motion.

Canonoid transformations preserve the Hamiltonian character of a specified dynamics, while allowing the geometric structure used to describe it to change \cite{CR88,CR89,CFR2013}. In the symplectic setting, they provide alternative
symplectic structures and can generate first integrals
through invariant endomorphisms and the traces of
their powers \cite{carinena2024geometric}. Contact canonoid transformations and associated trace
invariants were studied in \cite{azuaje2023canonical}, where the contact trace construction was considered under good contact Hamiltonian
hypotheses. Their relationship with scaling symmetries and
infinitesimal dynamical symmetries was developed in
\cite{azuaje2024scaling}. Related constructions for locally conformal
symplectic systems appear in \cite{azuaje2026canonical}, while singular fouling maps and their invariant tensors provide
another extension beyond invertible transformations
\cite{azuaje2026fouling}.

In this manuscript we formulate a pullback condition for possibly singular contact
canonoid maps and distinguish the dissipation rates of the original
contact form and its pullback. We construct invariants from a polynomial
of top-degree forms and from an endomorphism defined by normalized
one-forms. The latter construction requires neither invertibility of
the map nor preservation of either unnormalized contact form. The tensor is defined wherever the associated function
$K=-(\Phi^*\eta)(X_H)$ is nonzero, including points where
$H=0$. On $\{HK\neq0\}$, it admits an equivalent
characterization in terms of normalized one-forms.
The polynomial construction also applies without division by the
original Hamiltonian.

The two constructions are related: for a specific normalization, the
polynomial coefficients and the tensor traces determine one another
on their common domain. We prove this relation, give its
dimension-three specialization, and illustrate the results with
explicit mechanical systems. The singular-map definition concerns
identities on the original manifold; a Hamiltonian description on
the image is a separate question.

\section{Contact canonoid maps}
\label{sec:maps}

Let $(M^{2n+1},\eta,H)$ be a smooth autonomous contact Hamiltonian
system, with $n\geq1$. Denote the Reeb vector field of $\eta$ by $R$,
and adopt the convention
\begin{equation}
\eta(X_H)=-H,\qquad
\iota_{X_H}d\eta=dH-R(H)\eta.
\label{eq:convention}
\end{equation}
Write $X=X_H$ and $a=R(H)$. Then
\begin{equation}
\mathcal L_X\eta=-a\eta,\qquad X(H)=-aH.
\label{eq:original-rate}
\end{equation}
A function $I$ is conserved if $X(I)=0$, and a function $f$ is
dissipated if $X(f)=-af$. Thus, the quotient of two dissipated
quantities is conserved wherever its denominator is nonzero.
Conversely, $HI$ is dissipated whenever $I$ is conserved. We propose the following definition.

\begin{definition}
A smooth map $\Phi:M\to M$ is a \emph{contact canonoid map} for $(M,\eta,H)$ if there exists $b\in C^\infty(M)$
such that
\begin{equation}
\mathcal L_X(\Phi^*\eta)=-b\,\Phi^*\eta.
\label{eq:canonoid}
\end{equation}
A contact canonoid map that is a local diffeomorphism is called a
\emph{contact canonoid transformation}. The map is
\emph{dissipation-preserving} if
\begin{equation}
\mathcal L_X(\Phi^*\eta)=-a\,\Phi^*\eta.
\label{eq:equal-rates}
\end{equation}
These definitions also apply to maps between open subsets of $M$.
\end{definition}

Set
\begin{equation}
\alpha=\Phi^*\eta,\qquad K=-\alpha(X).
\label{eq:alpha-K}
\end{equation}
Cartan's formula shows that \eqref{eq:canonoid} is equivalent to
\begin{equation}
\alpha(X)=-K,\qquad
\iota_Xd\alpha=dK-b\alpha.
\label{eq:pullback-equations}
\end{equation}
If $\Phi$ is a local diffeomorphism, $\alpha$ is contact.
Evaluation on its Reeb vector field $R_\alpha$ gives
$b=R_\alpha(K)$, so $X$ is the contact Hamiltonian vector field of
$K$ relative to $\alpha$. On a domain where $\Phi$ is invertible,
\begin{equation}
\Phi_*X=X_{\widetilde H}^{\eta},
\qquad
\widetilde H=K\circ\Phi^{-1}.
\label{eq:target}
\end{equation}
Thus the definition recovers the contact canonoid condition
for local diffeomorphisms
\cite{azuaje2023canonical,azuaje2024scaling}.

\begin{remark}
For singular maps, $\alpha$ need not be contact. Indeed,
\[
\alpha\wedge(d\alpha)^n
=\Phi^*\bigl(\eta\wedge(d\eta)^n\bigr),
\]
so $\alpha$ is contact at a point precisely when $T\Phi$ is
invertible there. The definition requires no Reeb vector field
for $\alpha$. It includes maps with $\alpha=0$, for which $K=0$;
in general, the function $b$ need not be unique.

The pullback condition alone does not ensure that $X$ descends
to the image. Projectability requires a vector field
$\widetilde X$ satisfying
$T\Phi\circ X=\widetilde X\circ\Phi$; in particular, the
pushed-forward vectors must agree along each fiber.
Even when the image is a smooth submanifold and projectability
holds, a contact Hamiltonian description there requires an
appropriate contact form and a descended Hamiltonian.
No such descent is assumed in the results below.
\end{remark}

\begin{proposition}
\label{prop:balance}
For a contact canonoid map,
\begin{equation}
X(K)=-bK.
\label{eq:K-rate}
\end{equation}
On $\{H\neq0\}$,
\begin{equation}
X\left(\frac KH\right)=(a-b)\frac KH.
\label{eq:ratio-rate}
\end{equation}
In particular, a dissipation-preserving map gives a dissipated
quantity $K$ and a conserved quantity $K/H$.
\end{proposition}

\begin{proof}
Evaluate \eqref{eq:pullback-equations} on $X$:
\[
0=d\alpha(X,X)=X(K)-b\alpha(X)=X(K)+bK.
\]
Equation \eqref{eq:ratio-rate} follows from the quotient rule
and \eqref{eq:original-rate}.
\end{proof}

\section{Polynomial invariants and normalizations}
\label{sec:polynomial}

In this section we construct conserved quantities from one-forms having
the same dissipation rate as the original contact form
$\eta$. Such a one-form determines a polynomial in an
auxiliary parameter through a pencil of top-degree forms;
its coefficients are first integrals, and multiplication
by $H$ produces dissipated quantities. We then introduce
two normalizations of the pullback $\Phi^*\eta$ that make
this construction applicable when its dissipation rate
differs from that of $\eta$. The contact-volume
normalization is defined on the regular locus of $\Phi$,
whereas the Hamiltonian normalization requires only
$K\neq0$ and accommodates both singular points of the map
and zeros of the original Hamiltonian.

\subsection{One-forms with the same dissipation rate}

\begin{proposition}
\label{prop:pencil}
Let $\beta$ be a smooth one-form satisfying
\begin{equation}
\mathcal L_X\beta=-a\beta.
\label{eq:beta-rate}
\end{equation}
Then $J=-\beta(X)$ is dissipated. Let
$\mu=\eta\wedge(d\eta)^n$, and define $C_1,\ldots,C_{n+1}$
by
\begin{equation}
(\eta+t\beta)\wedge(d\eta+t\,d\beta)^n
=\left(1+\sum_{j=1}^{n+1}C_jt^j\right)\mu,
\label{eq:pencil}
\end{equation}
where $t$ is an auxiliary parameter. Then
\begin{equation}
X(C_j)=0,\qquad
X(HC_j)=-aHC_j,\qquad j=1,\ldots,n+1.
\label{eq:pencil-invariants}
\end{equation}
Neither $\beta$ nor every member of the pencil is required
to be contact.
\end{proposition}

\begin{proof}
Since $[X,X]=0$,
\[
X(J)=-(\mathcal L_X\beta)(X)=a\beta(X)=-aJ.
\]
Put $\eta_t=\eta+t\beta$ and $\mu_t=\eta_t\wedge(d\eta_t)^n$.
Equations \eqref{eq:original-rate} and \eqref{eq:beta-rate} give
\[
\mathcal L_X\eta_t=-a\eta_t,\qquad
\mathcal L_Xd\eta_t=-da\wedge\eta_t-a\,d\eta_t.
\]
Consequently,
\[
\mathcal L_X\mu_t=-(n+1)a\mu_t,
\]
because the term containing $da$ has two factors $\eta_t$.
In particular, $\mathcal L_X\mu=-(n+1)a\mu$.
Writing $\mu_t=\mathcal P(t)\mu$ yields $X(\mathcal P(t))=0$.
Comparing coefficients proves $X(C_j)=0$; multiplication
by $H$ gives the second identity in
\eqref{eq:pencil-invariants}.
\end{proof}

The first and last coefficients satisfy
\begin{align}
C_1\mu
&=\beta\wedge(d\eta)^n
+n\eta\wedge d\beta\wedge(d\eta)^{n-1},\\
C_{n+1}\mu&=\beta\wedge(d\beta)^n.
\label{eq:last-coefficient}
\end{align}

\begin{corollary}
\label{cor:equal-rates}
For a dissipation-preserving contact canonoid map,
take $\beta=\Phi^*\eta$ in Proposition~\ref{prop:pencil}.
Then $J=K$ is dissipated and the coefficients in
\eqref{eq:pencil} are conserved. This construction requires
no invertibility and no division by $H$, $K$, or a Jacobian.
\end{corollary}

\subsection{Contact-volume normalization}

\begin{proposition}
\label{prop:volume}
Let $\Phi$ be a contact canonoid transformation and define
the nonvanishing function $\rho$ by
\begin{equation}
\Phi^*\mu=\rho\mu.
\label{eq:rho}
\end{equation}
Then
\begin{equation}
\beta=|\rho|^{-1/(n+1)}\Phi^*\eta
\label{eq:volume-beta}
\end{equation}
satisfies \eqref{eq:beta-rate}. In particular,
\begin{equation}
J_\Phi=\frac{K}{|\rho|^{1/(n+1)}},
\qquad
I_\Phi=\frac{K}{H|\rho|^{1/(n+1)}}
\label{eq:volume-invariants}
\end{equation}
are respectively dissipated and conserved, with $I_\Phi$
defined on $\{H\neq0\}$. The polynomial coefficients
associated with \eqref{eq:volume-beta} are conserved
without requiring $H\neq0$.
\end{proposition}

\begin{proof}
The same calculation as in Proposition~\ref{prop:pencil}
gives
\[
\mathcal L_X\bigl(\alpha\wedge(d\alpha)^n\bigr)
=-(n+1)b\,\alpha\wedge(d\alpha)^n.
\]
Using $\alpha\wedge(d\alpha)^n=\rho\mu$, we obtain
\begin{equation}
X(\rho)=(n+1)(a-b)\rho.
\label{eq:rho-rate}
\end{equation}
Thus $c=|\rho|^{-1/(n+1)}$ satisfies $X(c)=(b-a)c$,
and
\[
\mathcal L_X(c\alpha)
=X(c)\alpha+c\mathcal L_X\alpha=-ac\alpha.
\]
Apply Proposition~\ref{prop:pencil} and divide the
dissipated quantity $J_\Phi$ by $H$.
\end{proof}

\begin{remark}
For Darboux coordinates with the same ordering convention
on source and target, $\rho=\det D\Phi$.
For a singular canonoid map, \eqref{eq:rho} still defines
a smooth function $\rho$, and \eqref{eq:rho-rate} remains
valid. Proposition~\ref{prop:volume} applies on
$\{\rho\neq0\}$, but its normalized expressions need not
extend across $\rho=0$.

For the normalization \eqref{eq:volume-beta},
\[
\beta\wedge(d\beta)^n
=\operatorname{sgn}(\rho)\mu.
\]
Hence its highest polynomial coefficient is
$C_{n+1}=\operatorname{sgn}(\rho)$ and is locally constant.
\end{remark}

\subsection{Hamiltonian normalization}

\begin{proposition}
\label{prop:ham-normalization}
For any contact canonoid map, the one-form
\begin{equation}
\beta=\frac HK\Phi^*\eta
\label{eq:ham-beta}
\end{equation}
satisfies \eqref{eq:beta-rate} on $\{K\neq0\}$.
Its coefficients in \eqref{eq:pencil} are therefore
conserved, including at singular points of $\Phi$.
No assumption $H\neq0$ is needed.
\end{proposition}

\begin{proof}
Equations \eqref{eq:original-rate} and \eqref{eq:K-rate}
give
\[
X\left(\frac HK\right)=(b-a)\frac HK.
\]
Therefore,
\[
\mathcal L_X\left(\frac HK\alpha\right)
=(b-a)\frac HK\alpha-b\frac HK\alpha
=-a\frac HK\alpha.
\]
Apply Proposition~\ref{prop:pencil}.
\end{proof}

Here $-\beta(X)=H$, so the scalar contraction itself
provides no additional quantity. The polynomial
coefficients may nevertheless be nonconstant.
More generally, any smooth function $c$ satisfying
$X(c)=(b-a)c$ produces a one-form $c\alpha$ obeying
\eqref{eq:beta-rate}.

\section{An invariant tensor and its polynomial invariants}
\label{sec:tensor}

We now develop a tensorial counterpart of the
Hamiltonian-normalized polynomial construction.
For a contact canonoid map, the pulled-back form
$\alpha=\Phi^*\eta$ and the associated function
$K=-\alpha(X_H)$ determine a smooth invariant endomorphism
on $\{K\neq0\}$. Its traces are conserved quantities,
and their products with $H$ are dissipated quantities.
The construction accommodates singular maps and remains
smooth at zeros of $H$, where the endomorphism vanishes.
We establish the precise relation between its traces and
the coefficients of the Hamiltonian-normalized polynomial,
showing that these constructions encode the same spectral
information and yield at most $n$ functionally independent
trace invariants in dimension $2n+1$. We conclude with an
explicit formula in dimension three.

\subsection{The intrinsic tensor}

\begin{proposition}
\label{prop:tensor}
Let $\Phi$ be a contact canonoid map for $(M,\eta,H)$, and set
\[
X=X_H,\qquad a=R(H),\qquad
\alpha=\Phi^*\eta,\qquad K=-\alpha(X).
\]
On the open set
\[
V=\{x\in M:K(x)\neq0\},
\]
define
\[
\sigma=-\frac{\alpha}{K},
\qquad
\mathcal D=\ker(\eta|_V).
\]
There exists a unique smooth endomorphism $N:TV\to TV$ such that
\[
\eta(NY)=0,
\qquad
d\eta(NY,Z)=-H\,d\sigma(Y,Z)
\]
for every vector field $Y$ on $V$ and every section $Z$ of
$\mathcal D$. Moreover,
\[
NX=0,\qquad \mathcal L_XN=0,
\]
and $N_x=0$ at every point $x\in V$ for which $H(x)=0$.

On the open subset
\[
U=\{x\in M:H(x)K(x)\neq0\},
\]
put $\lambda=-\eta/H$. Then $N$ is equivalently characterized by
\[
\lambda(NY)=0,
\qquad
d\lambda(NY,Z)=d\sigma(Y,Z)
\]
for all vector fields $Y,Z$ on $U$.

For every integer $k\geq1$, the functions
\[
I_k=\operatorname{tr}(N^k),
\qquad
D_k=H\operatorname{tr}(N^k)
\]
are smooth on $V$ and satisfy
\[
X(I_k)=0,\qquad X(D_k)=-aD_k.
\]
No invertibility of $\Phi$ is assumed.
\end{proposition}

\begin{proof}
The canonoid condition and the balance law for $K$ give
\[
\mathcal L_X\alpha=-b\alpha,
\qquad
X(K)=-bK.
\]
Consequently,
\[
\mathcal L_X\sigma=0,
\qquad
\sigma(X)=1.
\]
Cartan's formula therefore yields
\[
\iota_Xd\sigma=0.
\]

Since $\eta$ is contact, the restriction of $d\eta$ to
$\mathcal D$ is nondegenerate. Thus, for each vector field $Y$,
there is a unique section $NY$ of $\mathcal D$ satisfying
\[
d\eta(NY,Z)=-H\,d\sigma(Y,Z),
\qquad Z\in\Gamma(\mathcal D).
\]
The inverse of the corresponding bundle isomorphism
$\mathcal D\to\mathcal D^*$ is smooth, so this construction
defines a smooth endomorphism $N$. Taking $Y=X$ gives $NX=0$.
At a point where $H=0$, nondegeneracy gives $N=0$.

We next prove invariance. Since
\[
\mathcal L_X\eta=-a\eta,
\]
the distribution $\mathcal D$ is invariant under $X$.
Lie differentiation of $\eta(NY)=0$ gives
\[
\eta\bigl((\mathcal L_XN)Y\bigr)=0.
\]
Furthermore,
\[
\mathcal L_Xd\eta=-da\wedge\eta-a\,d\eta,
\qquad
\mathcal L_Xd\sigma=0,
\qquad
X(H)=-aH.
\]
We may therefore Lie differentiate the defining identity for
$N$, with its second argument restricted to $\mathcal D$.
For $Z\in\Gamma(\mathcal D)$ this gives
\[
(\mathcal L_Xd\eta)(NY,Z)
+d\eta\bigl((\mathcal L_XN)Y,Z\bigr)
=aH\,d\sigma(Y,Z).
\]
Both $NY$ and $Z$ lie in $\mathcal D$, so
\[
(\mathcal L_Xd\eta)(NY,Z)
=-a\,d\eta(NY,Z)
=aH\,d\sigma(Y,Z).
\]
Hence
\[
d\eta\bigl((\mathcal L_XN)Y,Z\bigr)=0
\qquad\text{for every }Z\in\Gamma(\mathcal D).
\]
Together with
$\eta((\mathcal L_XN)Y)=0$, nondegeneracy on $\mathcal D$
implies $\mathcal L_XN=0$.

On $U$, the form $\lambda=-\eta/H$ satisfies
\[
\lambda(X)=1,\qquad
\mathcal L_X\lambda=0,\qquad
\iota_Xd\lambda=0.
\]
Since $\ker\lambda=\mathcal D|_U$, we have
\[
TU=\mathcal D|_U\oplus\operatorname{span}\{X\}.
\]
For $Z\in\Gamma(\mathcal D|_U)$, the identity
\[
d\lambda=-\frac{1}{H}d\eta
+\frac{1}{H^2}dH\wedge\eta
\]
and the condition $\eta(NY)=0$ give
\[
d\lambda(NY,Z)
=-\frac{1}{H}d\eta(NY,Z)
=d\sigma(Y,Z).
\]
Both sides vanish when $Z=X$, so the same identity holds
for arbitrary $Z$. Conversely, restricting this identity
to $Z\in\Gamma(\mathcal D|_U)$ recovers the defining
equation for $N$.

Finally, Lie differentiation is a derivation for composition
and commutes with contraction. Therefore,
\[
X\bigl(\operatorname{tr}(N^k)\bigr)
=\operatorname{tr}\bigl(\mathcal L_X(N^k)\bigr)=0.
\]
Multiplication by $H$, together with $X(H)=-aH$, proves
the balance law for $D_k$.
\end{proof}

\begin{remark}
These constructions do not by themselves establish
nonconstancy, functional independence, or involutivity
of the resulting invariants. In particular, no
Nijenhuis condition is needed for
$\mathcal L_XN=0$ or conservation of its traces.

The normalization is unchanged if the same vector
field is represented by
$\widehat\eta=f\eta$ and $\widehat H=fH$, with $f$
nowhere zero. Indeed,
$\widehat\alpha=(f\circ\Phi)\alpha$ and
$\widehat K=(f\circ\Phi)K$, so $\lambda$, $\sigma$,
and $N$ remain unchanged.
\end{remark}

\subsection{Relation between the polynomial coefficients and traces}

\begin{proposition}
\label{prop:spectral}
Use the normalization $\beta=(H/K)\alpha$ from
Proposition~\ref{prop:ham-normalization}, and write
\[
\mathcal P(t)=1+\sum_{j=1}^{n+1}C_jt^j
\]
for its polynomial in \eqref{eq:pencil}.
On $V=\{K\neq0\}$, let
$\mathcal D=\ker(\eta|_V)$ and $A=N|_{\mathcal D}$. There exists a polynomial
$Q(t)=1+q_1t+\cdots+q_nt^n$ such that
\begin{equation}
\mathcal P(t)=(1+t)Q(t),\qquad
\det(\id_{\mathcal D}+tA)=Q(t)^2.
\label{eq:spectral-relation}
\end{equation}
The coefficients of $Q$ and the first $n$ traces of
$N$ determine one another. In particular,
\begin{align}
\tr N&=2(C_1-1),\\
\tr(N^2)
&=2(C_1-1)^2-4(C_2-C_1+1).
\label{eq:first-trace-relations}
\end{align}
All higher traces are functions of $q_1,\ldots,q_n$.
Thus a single tensor $N$ gives at most $n$
functionally independent trace invariants.
\end{proposition}

\begin{proof}
First work on $U=\lbrace HK\neq 0\rbrace$.
Put
\[
\omega=d\lambda|_{\mathcal D},\qquad
\vartheta=d\sigma|_{\mathcal D},
\]
and define $Q(t)$ by
\[
(\omega+t\vartheta)^n=Q(t)\omega^n.
\]
This is a smooth polynomial in $t$, with $Q(0)=1$,
because $\omega$ is nondegenerate.

Since $\eta+t\beta=-H(\lambda+t\sigma)$,
\[
(\eta+t\beta)\wedge(d\eta+t\,d\beta)^n
=(-H)^{n+1}
(\lambda+t\sigma)\wedge(d\lambda+t\,d\sigma)^n.
\]
The terms containing $dH$ vanish in this wedge product.
Also,
\[
(\lambda+t\sigma)(X)=1+t,\qquad
\iota_X(d\lambda+t\,d\sigma)=0.
\]
Contract with $X$, restrict to $\mathcal D$, and compare
with $\mu=(-H)^{n+1}\lambda\wedge(d\lambda)^n$.
This proves $\mathcal P(t)=(1+t)Q(t)$.

In a local frame of $\mathcal D$, let $\Omega$ and
$\Theta$ be the matrices of $\omega$ and $\vartheta$.
The defining equation for $A$ is
$A^{\mathsf T}\Omega=\Theta$, so
\[
\Omega+t\Theta=(\id+tA)^{\mathsf T}\Omega.
\]
The determinant of a skew-symmetric matrix is the
square of its Pfaffian. Therefore,
\[
\det(\id+tA)
=\frac{\det(\Omega+t\Theta)}{\det\Omega}
=Q(t)^2.
\]
As formal power series at $t=0$,
\[
2\log Q(t)
=\sum_{k\geq1}\frac{(-1)^{k-1}}{k}
\tr(A^k)t^k.
\]
Since $NX=0$, $\tr(A^k)=\tr(N^k)$.
Comparison of coefficients proves the assertions.
In particular, $q_1=C_1-1$ and
$q_2=C_2-C_1+1$, with $q_2=0$ if $n=1$.

To extend the identities to $V$, set $q_0=1$ and define
\[
q_j=C_j-q_{j-1},\qquad j=1,\ldots,n.
\]
These functions are smooth on $V$ and agree on $U$ with
the coefficients obtained above. Moreover, $U$ is dense
in $V$: if $H$ vanished identically on a nonempty open
subset of $V$, then $X_H$ and hence
$K=-\alpha(X_H)$ would vanish there, a contradiction.
Since $N$ is smooth on $V$ by Proposition~5, all the
polynomial, determinant, and trace identities extend
from $U$ to $V$ by continuity.
\end{proof}

Thus Propositions~\ref{prop:ham-normalization}
and~\ref{prop:tensor} encode the same spectral
invariants throughout $\{K\neq0\}$, including points
where $H=0$. This equivalence concerns the normalization
$\beta=(H/K)\alpha$, rather than an arbitrary one-form
in Proposition~\ref{prop:pencil}.

\subsection{The dimension-three formula}

\begin{corollary}
\label{cor:three}
Suppose $\dim M=3$, and define $\rho$ by
$\Phi^*\mu=\rho\mu$, allowing $\rho$ to vanish.
On $U=\{HK\neq0\}$, put
\[
\Pi=\id_{TU}-X\otimes\lambda,\qquad
\kappa=\rho\left(\frac HK\right)^2.
\]
Then
\begin{equation}
N=\kappa \Pi,\qquad
\tr(N^k)=2\kappa^k,\qquad k\geq1.
\label{eq:three-N}
\end{equation}
In particular, these traces provide at most one
functionally independent invariant.
\end{corollary}

\begin{proof}
In dimension three,
\[
\lambda\wedge d\lambda=\frac{\mu}{H^2},
\qquad
\sigma\wedge d\sigma=\frac{\rho\mu}{K^2}.
\]
Contracting with $X$ gives
\[
d\lambda=\frac{\iota_X\mu}{H^2},
\qquad
d\sigma=\frac{\rho\,\iota_X\mu}{K^2}
=\kappa\,d\lambda.
\]
The projection $\Pi$ takes values in $\ker\lambda$ and
satisfies $d\lambda(\Pi Y,Z)=d\lambda(Y,Z)$.
Uniqueness in Proposition~\ref{prop:tensor} gives
$N=\kappa \Pi$. Since $\Pi^2=\Pi$ and $\tr \Pi=2$,
\eqref{eq:three-N} follows.
\end{proof}

Expanding the projection gives
\[
N=\frac{\rho H^2}{K^2}\,\mathrm{id}
  +\frac{\rho H}{K^2}\,X\otimes\eta.
\]
This expression is smooth on $V=\{K\neq0\}$ and agrees
there with the tensor of Proposition~\ref{prop:tensor}, by continuity
from the dense subset $U$. Consequently,
\[
\operatorname{tr}(N^k)
=2\left(\frac{\rho H^2}{K^2}\right)^k
\]
holds throughout $V$.

\section{Construction and examples}
\label{sec:examples}

\subsection{Shears of the contact coordinate}

\begin{proposition}
\label{prop:shear}
In Darboux coordinates $(q^i,p_i,z)$, let
\[
\eta=dz-p_i\,dq^i,\qquad
H=h(q,p)+\gamma z,
\]
where $\gamma$ is constant. For a smooth function
$f(q,p)$, consider
\[
\Phi_f(q,p,z)=(q,p,z+f(q,p)).
\]
This shear is a dissipation-preserving contact canonoid
transformation if and only if
\begin{equation}
d\bigl(X_H(f)+\gamma f\bigr)=0.
\label{eq:shear-condition}
\end{equation}
On a connected domain, this is equivalent to
$X_H(f)+\gamma f=C$ for a constant $C$, and then
\begin{equation}
K=H+\gamma f-C.
\label{eq:shear-K}
\end{equation}
\end{proposition}

\begin{proof}
The inverse shear replaces $f$ by $-f$, and
$\Phi_f^*\eta=\eta+df$. Since
$\mathcal L_{X_H}\eta=-\gamma\eta$,
\[
\mathcal L_{X_H}(\Phi_f^*\eta)
+\gamma\Phi_f^*\eta
=d\bigl(X_H(f)+\gamma f\bigr).
\]
This proves \eqref{eq:shear-condition}.
Finally, $K=-(\eta+df)(X_H)=H-X_H(f)$,
which gives \eqref{eq:shear-K}.
\end{proof}

Equation \eqref{eq:shear-condition} provides a direct
way to construct maps: choose an ansatz for $f$ and
solve the resulting equations for its coefficients.

\subsection{Singular maps for the damped free particle}

Consider
\[
\eta=dz-p\,dq,\qquad
H=\frac{p^2}{2}+\gamma z,\qquad \gamma>0.
\]
Its Hamiltonian vector field is
\[
X=p\partial_q-\gamma p\partial_p
+\left(\frac{p^2}{2}-\gamma z\right)\partial_z.
\]
The Darboux coordinates
\begin{equation}
u=q+\frac p\gamma,\qquad
w=z+\frac{p^2}{2\gamma}
\label{eq:free-coordinates}
\end{equation}
give
\begin{equation}
\eta=dw-p\,du,\qquad
H=\gamma w,\qquad
X=-\gamma p\partial_p-\gamma w\partial_w.
\label{eq:free-adapted}
\end{equation}

\begin{example}[A polynomial invariant from a rank-two map]
The map
\[
\Phi(u,p,w)=(u,0,e^u w)
\]
has rank two everywhere and satisfies
\[
\alpha=d(e^u w),\qquad
\mathcal L_X\alpha=-\gamma\alpha,\qquad
K=e^uH.
\]
Since $d\alpha=0$,
\[
(\eta+t\alpha)\wedge d(\eta+t\alpha)
=(1+te^u)\mu.
\]
Thus
\[
C_1=e^{q+p/\gamma},\qquad
K=e^{q+p/\gamma}\left(\frac{p^2}{2}+\gamma z\right)
\]
are respectively conserved and dissipated on all of
$\mathbb R^3$, including $H=0$.
Here $\rho=0$, so contact-volume normalization is
unavailable. On $w\neq0$, Corollary~\ref{cor:three}
gives $N=0$. This illustrates that the unnormalized
polynomial construction can retain information
not present in the normalized tensor.
\end{example}

\begin{example}[A tensor smooth across the singular locus]
For the same system, let
\[
\Phi(u,p,w)=(u,up,w).
\]
Then
\[
\alpha=dw-up\,du,\qquad
\mathcal L_X\alpha=-\gamma\alpha,\qquad
K=H,\qquad \rho=u.
\]
Thus $\Phi$ is singular on $u=0$.
On $U=\{w\neq0\}$, Corollary~\ref{cor:three} yields
\[
N=u\Pi,\qquad
\Pi=\id-X\otimes\lambda,\qquad
\lambda=\frac{p\,du-dw}{\gamma w}.
\]
In the coordinate basis $(\partial_u,\partial_p,\partial_w)$,
\begin{equation}
[N]=u
\begin{pmatrix}
1&0&0\\
p^2/w&1&-p/w\\
p&0&0
\end{pmatrix}.
\label{eq:free-N}
\end{equation}
Its eigenvalues are $0,u,u$, so
\[
I_k=2\left(q+\frac p\gamma\right)^k,\qquad
D_k=2\left(\frac{p^2}{2}+\gamma z\right)
\left(q+\frac p\gamma\right)^k.
\]
These scalar functions extend to all of $\mathbb R^3$
and obey $X(I_k)=0$, $X(D_k)=-\gamma D_k$.
The tensor itself is smooth at $u=0$ whenever $w\neq0$.

The polynomial construction gives
\[
(\eta+t\alpha)\wedge d(\eta+t\alpha)
=(1+t)(1+tu)\mu,
\]
so $C_1=1+u$ and $C_2=u$, in agreement with
Proposition~\ref{prop:spectral}.
\end{example}

\subsection{A transformation with unequal dissipation rates}

\begin{example}
For \eqref{eq:free-adapted}, restrict to $p>0,w>0$
and take
\[
\Phi(u,p,w)=(u,p^2,w^2).
\]
This is a diffeomorphism of that domain onto itself.
Its pullback and associated Hamiltonian are
\[
\alpha=2w\,dw-p^2du,\qquad K=2\gamma w^2.
\]
A direct calculation gives
\[
\mathcal L_X\alpha=-2\gamma\alpha,\qquad
a=\gamma,\quad b=2\gamma,\qquad
\rho=4pw.
\]
Although $K/H=2w$ is not conserved,
Proposition~\ref{prop:volume} gives
\[
I_\Phi=\frac{K}{H\sqrt{\rho}}
=\sqrt{\frac wp},
\qquad
J_\Phi=\gamma w\sqrt{\frac wp}.
\]
They satisfy $X(I_\Phi)=0$ and
$X(J_\Phi)=-\gamma J_\Phi$.

The Hamiltonian normalization is
\[
\beta=\frac HK\alpha
=dw-\frac{p^2}{2w}\,du,
\]
and its polynomial is
\[
\mathcal P(t)=(1+t)\left(1+t\frac pw\right).
\]
Consistently, Corollary~\ref{cor:three} gives
\[
N=\frac pw \Pi,\qquad
\tr(N^k)=2\left(\frac pw\right)^k.
\]
\end{example}

\subsection{The damped harmonic oscillator}

\begin{example}
Consider the contact formulation of the damped
harmonic oscillator \cite{Bravetti2017,GGMRR2020}:
\[
\eta=dz-y\,dx,\qquad
H=\frac{y^2}{2m}+\frac{m\omega^2x^2}{2}+\gamma z,
\qquad m,\omega,\gamma>0.
\]
Then
\[
X=\frac ym\partial_x-(m\omega^2x+\gamma y)\partial_y
+\left(\frac{y^2}{2m}-\frac{m\omega^2x^2}{2}
-\gamma z\right)\partial_z.
\]
A homogeneous quadratic ansatz for a function satisfying
$X(F)=-\gamma F$ gives, up to a constant factor,
\[
F=\frac{y^2}{2m}+\frac{m\omega^2x^2}{2}
+\frac{\gamma}{2}xy.
\]
Indeed, if $E=y^2/(2m)+m\omega^2x^2/2$, then
\[
X(E)=-\frac{\gamma}{m}y^2,\qquad
X(xy)=\frac{y^2}{m}-m\omega^2x^2-\gamma xy,
\]
and hence $X(F)=-\gamma F$.

Proposition~\ref{prop:shear}, with $f=F/\gamma$,
gives the global canonoid transformation
\[
\Phi(x,y,z)=\left(x,y,z+\frac F\gamma\right),
\qquad
\alpha=\eta+\frac{dF}{\gamma},
\qquad
K=H+F.
\]
Here $\rho=1$ and $X(K)=-\gamma K$.
On $U=\{H(H+F)\neq0\}$,
\begin{equation}
N=\left(\frac{H}{H+F}\right)^2
\left(\id+\frac1H X\otimes\eta\right).
\label{eq:oscillator-N}
\end{equation}
In the coordinate basis $(\partial_x,\partial_y,\partial_z)$,
\[
[N]=\frac{H}{K^2}
\begin{pmatrix}
H-y^2/m&0&y/m\\
y(m\omega^2x+\gamma y)&H&-m\omega^2x-\gamma y\\
y(H-y^2/m)&0&y^2/m
\end{pmatrix}.
\]
Its eigenvalues are $0,(H/K)^2,(H/K)^2$, and
\begin{equation}
I_k=2\left(\frac{H}{H+F}\right)^{2k},
\qquad
D_k=\frac{2H^{2k+1}}{(H+F)^{2k}}.
\label{eq:oscillator-invariants}
\end{equation}
All traces are functions of the first integral
$J=F/H$, since $I_k=2(1+J)^{-2k}$.
They supply one functionally independent invariant
on an open dense subset of $U$.

These formulas apply in the underdamped, critically
damped, and overdamped regimes. The expression
\[
N=\frac{H^2}{K^2}\id+\frac{H}{K^2}X\otimes\eta
\]
extends smoothly to $K\neq0$, including $H=0$.
Its invariance and the balance laws for
\eqref{eq:oscillator-invariants} extend by smoothness
from the dense subset $U$.
\end{example}

\subsection{Gravity with linear friction}

\begin{example}
Consider the system discussed in
\cite[Section~5.2]{GGMRR2020}, with
\[
\eta=dz-p_x\,dx-p_y\,dy,\qquad
H=\frac{p_x^2+p_y^2}{2m}+mg_0y+\gamma z,
\]
where $m>0$, $\gamma>0$, and $g_0\in\mathbb R$.
Its vector field is
\begin{align*}
X={}&\frac{p_x}{m}\partial_x+\frac{p_y}{m}\partial_y
-\gamma p_x\partial_{p_x}
-(mg_0+\gamma p_y)\partial_{p_y}\\
&+\left(\frac{p_x^2+p_y^2}{2m}-mg_0y-\gamma z\right)
\partial_z.
\end{align*}
Put $A=mg_0+\gamma p_y$, so $X(A)=-\gamma A$.
For a nonzero constant $\varepsilon$, the shear
\[
\Phi_\varepsilon(x,y,p_x,p_y,z)
=(x,y,p_x,p_y,z+\varepsilon p_y)
\]
satisfies Proposition~\ref{prop:shear}, because
\[
X(\varepsilon p_y)+\gamma\varepsilon p_y
=-\varepsilon mg_0.
\]
Consequently,
\[
\alpha=\eta+\varepsilon\,dp_y,\qquad
K=H+\varepsilon A,\qquad
\mathcal L_X\alpha=-\gamma\alpha.
\]
For $\varepsilon\neq0$, $\alpha$ is not proportional
to $\eta$. The associated conserved quantity is
\begin{equation}
J_y=\frac{K-H}{\varepsilon H}
=\frac{mg_0+\gamma p_y}{H}.
\label{eq:gravity-Jy}
\end{equation}
Both $A$ and $K$ are globally dissipated.
In the image coordinates, the transformed Hamiltonian is
\[
\widetilde H=K\circ\Phi_\varepsilon^{-1}
=\frac{\widetilde p_x^2+\widetilde p_y^2}{2m}
+mg_0\widetilde y+\gamma\widetilde z+\varepsilon mg_0.
\]

To compute $N$ on $U=\{HK\neq0\}$, set
\[
r=\frac HK,\qquad c=\frac{\varepsilon H}{K^2},
\qquad \Pi=\id-X\otimes\lambda,
\]
and define
\[
E_x=\partial_x+p_x\partial_z,\qquad
E_y=\partial_y+p_y\partial_z,\qquad
W=X+H\partial_z.
\]
These fields lie in $\mathcal D=\ker\eta$, and
\[
W=\frac{p_x}{m}E_x+\frac{p_y}{m}E_y
-\gamma p_x\partial_{p_x}-A\partial_{p_y},
\]
with
$\iota_Wd\eta=dH-\gamma\eta$ and
$\iota_{E_y}d\eta=dp_y$.
Since $d\alpha=d\eta$, restriction to
$\mathcal D\times\mathcal D$ gives
\[
d\lambda|_{\mathcal D}=-\frac1H d\eta|_{\mathcal D},
\qquad
d\sigma|_{\mathcal D}
=-\frac1K d\eta|_{\mathcal D}
+\frac{\varepsilon}{K^2}(dH\wedge dp_y)|_{\mathcal D}.
\]
It follows that
\begin{equation}
N=r\Pi+c\left[
W\otimes(dp_y\circ \Pi)-E_y\otimes(dH\circ \Pi)
\right],
\label{eq:gravity-N}
\end{equation}
where
\[
dp_y\circ \Pi=dp_y-\frac AH\eta,\qquad
dH\circ \Pi=dH-\gamma\eta.
\]
Indeed, \eqref{eq:gravity-N} takes values in
$\mathcal D$, annihilates $X$, and for
$Y,Z\in\mathcal D$ satisfies
\[
d\lambda(NY,Z)
=-\frac1K d\eta(Y,Z)
+\frac{\varepsilon}{K^2}
\bigl[dH(Y)dp_y(Z)-dp_y(Y)dH(Z)\bigr]
=d\sigma(Y,Z).
\]
Since both two-forms annihilate $X$, the identity
holds for arbitrary $Y,Z$.

In the ordered frame
$\mathcal B=(X,E_y,E_x,\partial_{p_x},\partial_{p_y})$,
using $r-cA=r^2$, we obtain
\begin{equation}
[N]_{\mathcal B}=
\begin{pmatrix}
0&0&0&0&0\\
0&r^2&-c\gamma p_x&-cp_x/m&0\\
0&0&r&0&cp_x/m\\
0&0&0&r&-c\gamma p_x\\
0&0&0&0&r^2
\end{pmatrix}.
\label{eq:gravity-matrix}
\end{equation}
This frame is well defined on $U$ because
$\eta(X)=-H\neq0$. Hence
\[
\det(t\id-N)=t(t-r)^2(t-r^2)^2,
\]
and
\begin{equation}
I_k=2r^k+2r^{2k},\qquad
D_k=2Hr^k+2Hr^{2k}
\label{eq:gravity-invariants}
\end{equation}
are respectively conserved and dissipated.
The formulas do not require diagonalizability.
For $\beta=(H/K)\alpha$, the polynomial in
Proposition~\ref{prop:spectral} is
\[
\mathcal P(t)=(1+t)(1+rt)(1+r^2t).
\]
All these spectral invariants depend on $J_y$ through
$r=(1+\varepsilon J_y)^{-1}$.

The horizontal shear
$\Psi_\varepsilon(x,y,p_x,p_y,z)
=(x,y,p_x,p_y,z+\varepsilon p_x)$ gives
\[
K_x=H+\varepsilon\gamma p_x,\qquad
J_x=\frac{K_x-H}{\varepsilon\gamma H}=\frac{p_x}{H}.
\]
The two conserved quantities $J_x,J_y$ are
functionally independent on an open dense subset
of $\{H\neq0\}$, since
\[
\frac{\partial(J_x,J_y)}{\partial(p_x,z)}
=-\frac{\gamma A}{H^3}\neq0
\quad\text{when }A\neq0.
\]
The reciprocal $H/p_x$, defined where $p_x\neq0$,
is the conserved quantity discussed in
\cite[Section~5.2]{GGMRR2020}.
\end{example}

\subsection{Two independent trace invariants in dimension five}
\begin{example}
Consider the two-dimensional damped free particle,
\[
\eta=dz-p_1\,dq^1-p_2\,dq^2,
\qquad
H=\frac{p_1^2+p_2^2}{2}+\gamma z,
\qquad \gamma>0.
\]
The change of coordinates
\[
u_i=q^i+\frac{p_i}{\gamma},
\qquad
w=z+\frac{p_1^2+p_2^2}{2\gamma}
\]
gives
\[
\eta=dw-p_1\,du_1-p_2\,du_2,
\qquad
H=\gamma w,
\qquad
X=-\gamma p_1\partial_{p_1}
  -\gamma p_2\partial_{p_2}
  -\gamma w\partial_w.
\]

Define
\[
\Phi(u_1,u_2,p_1,p_2,w)
=(u_1,u_2,u_1p_1,u_2p_2,w).
\]
Its pullback is
\[
\alpha=\Phi^*\eta
=dw-u_1p_1\,du_1-u_2p_2\,du_2,
\]
and direct calculation yields
\[
\mathcal L_X\alpha=-\gamma\alpha,
\qquad
K=-\alpha(X)=\gamma w=H.
\]
Thus $\Phi$ is a dissipation-preserving contact canonoid map.
Writing $\mu=\eta\wedge(d\eta)^2$, we have
\[
\Phi^*\mu=u_1u_2\,\mu,
\qquad
\rho=u_1u_2.
\]
The map has rank five where $u_1u_2\neq0$, rank four where
exactly one of $u_1,u_2$ vanishes, and rank three where
$u_1=u_2=0$.

The polynomial construction is defined on all of
$\mathbb R^5$. Indeed,
\[
\eta+t\alpha
=(1+t)\,dw
 -(1+tu_1)p_1\,du_1
 -(1+tu_2)p_2\,du_2,
\]
so
\[
(\eta+t\alpha)\wedge(d\eta+t\,d\alpha)^2
=(1+t)(1+tu_1)(1+tu_2)\,\mu.
\]
Consequently,
\[
C_1=1+u_1+u_2,\qquad
C_2=u_1+u_2+u_1u_2,\qquad
C_3=u_1u_2
\]
are globally defined first integrals.

On $U=\{w\neq0\}$, put
\[
\lambda=-\frac{\eta}{H},
\qquad
\sigma=-\frac{\alpha}{H}.
\]
In the coordinate basis
$(\partial_{u_1},\partial_{u_2},
  \partial_{p_1},\partial_{p_2},\partial_w)$,
the invariant tensor is
\[
[N]=
\begin{pmatrix}
u_1 & 0 & 0 & 0 & 0\\
0 & u_2 & 0 & 0 & 0\\
\dfrac{u_1p_1^2}{w}
&
\dfrac{u_2p_1p_2}{w}
&
u_1 & 0 & -\dfrac{u_1p_1}{w}\\
\dfrac{u_1p_1p_2}{w}
&
\dfrac{u_2p_2^2}{w}
&
0 & u_2 & -\dfrac{u_2p_2}{w}\\
u_1p_1 & u_2p_2 & 0 & 0 & 0
\end{pmatrix}.
\]
This expression is smooth across the singular locus
$\{u_1u_2=0\}\cap U$.

To verify the formula and identify its spectrum, define
\[
E_i=\partial_{u_i}+p_i\partial_w,
\qquad
c=\frac{(u_1-u_2)p_1p_2}{w}.
\]
The fields $E_1,E_2,\partial_{p_1},\partial_{p_2}$ span
$\mathcal D=\ker\eta$. On this distribution,
\[
d\lambda|_{\mathcal D}
=-\frac{1}{\gamma w}
 \bigl(du_1\wedge dp_1+du_2\wedge dp_2\bigr)|_{\mathcal D},
\]
and
\[
d\sigma|_{\mathcal D}
=-\frac{1}{\gamma w}
 \bigl(
 u_1\,du_1\wedge dp_1
 +u_2\,du_2\wedge dp_2
 -c\,du_1\wedge du_2
 \bigr)|_{\mathcal D}.
\]
The displayed tensor satisfies
\[
NX=0,\qquad
NE_1=u_1E_1+c\,\partial_{p_2},
\qquad
NE_2=u_2E_2-c\,\partial_{p_1},
\]
and
\[
N\partial_{p_1}=u_1\partial_{p_1},
\qquad
N\partial_{p_2}=u_2\partial_{p_2}.
\]
These identities verify its defining equations. Since
$\eta(X)=-\gamma w\neq0$ on $U$, the ordered frame
\[
\mathcal B=(X,E_1,E_2,\partial_{p_1},\partial_{p_2})
\]
is well defined there, and
\[
[N]_{\mathcal B}
=
\begin{pmatrix}
0&0&0&0&0\\
0&u_1&0&0&0\\
0&0&u_2&0&0\\
0&0&-c&u_1&0\\
0&c&0&0&u_2
\end{pmatrix}.
\]
Therefore,
\[
\det(s\,\mathrm{id}-N)
=s(s-u_1)^2(s-u_2)^2,
\]
and, for every $k\geq1$,
\[
I_k=\operatorname{tr}(N^k)
=2(u_1^k+u_2^k),
\qquad
D_k=2\gamma w(u_1^k+u_2^k).
\]
These scalar functions extend smoothly to all of
$\mathbb R^5$ and satisfy
\[
X(I_k)=0,\qquad X(D_k)=-\gamma D_k.
\]

In particular,
\[
I_1=2(u_1+u_2),
\qquad
I_2=2(u_1^2+u_2^2),
\]
and
\[
dI_1\wedge dI_2
=8(u_2-u_1)\,du_1\wedge du_2.
\]
Thus a single tensor produces two functionally independent
trace invariants on the open dense subset
$U\cap\{u_1\neq u_2\}$, attaining the upper bound for $n=2$.
This independence also holds at singular points of $\Phi$
where exactly one of $u_1,u_2$ vanishes.
\end{example}

\section{Concluding remarks}

The pullback definition permits conserved and dissipated
quantities to be associated with contact canonoid maps
even when the maps are singular. Polynomial invariants
require only a one-form with the appropriate dissipation
rate. The invariant tensor is smooth wherever the
associated Hamiltonian $K$ is nonzero and vanishes at
points where $H=0$. For the Hamiltonian-normalized pencil,
its coefficients and the tensor traces encode the same
spectral information throughout $\{K\neq0\}$.

The examples distinguish the roles of singularity, normalization,
and unequal dissipation rates. The shear condition also gives a
practical starting point for constructing maps through polynomial
ansatzes. Determining when the resulting invariants are independent,
when they satisfy suitable involutivity conditions, and when
singular maps admit contact Hamiltonian descriptions on their images
requires additional hypotheses.

\section*{Acknowledgements}

This research was supported by the European Union and the Czech Ministry of Education under project CZ.02.01.01/00/22\_011/0008569 "Czech Technical University - International Postdoc Programme CROP".

\section*{AI Use Statement}
The author acknowledges the use of Artificial intelligence (AI)-assisted tools for the preparation of this manuscript. All results obtained with the assistance of these tools were independently reviewed and verified by the author. The author takes full responsibility for the accuracy, originality, and integrity of the final manuscript.

\bibliography{refs} 
\bibliographystyle{unsrt} 

\end{document}